\documentclass[runningheads]{llncs}
\usepackage[T1]{fontenc}
\usepackage[ruled,vlined]{algorithm2e}

\usepackage{boxedminipage}
\usepackage[most]{tcolorbox}
\usetikzlibrary{shadows}

\usepackage{amsfonts}

\usepackage{tikz}
\usetikzlibrary{shapes.geometric}
\usetikzlibrary{arrows.meta}
\usetikzlibrary{positioning}
\usetikzlibrary{arrows}

\usepackage{indentfirst}

\pdfoutput=1 

\newcommand{\br}[1]{\mathopen{}\left( #1 \right)}
\newcommand{\brc}[1]{\mathopen{}\left\{ #1 \right\}}

\newcommand{\fl}[1]{\mathopen{}\left\lfloor #1 \right\rfloor}

\newcommand{\angl}[1]{\mathopen{}\langle #1 \rangle}

\newcommand{\diam}{\operatorname{diam}}

\newcommand{\OPT}{\texttt{OPT}}

\SetKwFunction{FSeparator}{Separator} 
\SetKwFunction{FSeparatorFPTAS}{SeparatorFPTAS}
\SetKwFunction{FAlgorithmMinCut}{AlgorithmMinCut}

\begin{document}
\title{Approximating the Average-Case Graph Search Problem with Non-Uniform Costs}
\titlerunning{Average-Case Graph Search Problem with Non-Uniform Costs}
%
\author{Michał Szyfelbein\orcidID{0009-0009-9894-9671}}
\authorrunning{M. Szyfelbein}
%
\institute{Gdańsk University of Technology, Gdańsk, Poland}
\maketitle              
\begin{abstract}
Consider the following generalization of the classic binary search problem: A searcher is required to find a hidden target vertex $x$ in a graph $G$. To do so, they iteratively perform queries to an oracle, each about a chosen vertex $v$. After each such call, the oracle responds whether the target was found and if not, the searcher receives as a reply the connected component in $G-v$ which contains $x$. Additionally, each vertex $v$ may have a different query cost $c(v)$ and a different weight $w(v)$. The goal is to find the optimal querying strategy which minimizes the weighted average-case cost required to find $x$. The problem is NP-hard even for uniform weights and query costs. Inspired by the progress on the edge query variant of the problem [SODA '17], we establish a connection between searching and vertex separation. By doing so, we provide an $O(\sqrt{\log n})$-approximation algorithm for general graphs and a $(4+\epsilon)$-approximation algorithm for the case when the input is a tree.

\end{abstract}
\section{Introduction}
\label{sec:typesetting-summary}

Searching in graph structures is a fundamental problem in computer science, with applications ranging from machine learning to operations research. The input graph $G$ is assumed to contain a hidden \textit{target} element $x$ which the searcher is required to locate. In the classical setting, while exploring the graph, the searcher usually obtains only a local information about the placement of $x$. This principle underlines the well-known BFS and DFS strategies, which allow the searcher to locate the target when the only information received, while visiting a node, is whether it is the target. We study a more global search model in which the searcher is allowed to perform arbitrary \textit{queries}, each about a chosen vertex $v$. A \textit{response} to such query consists of information whether $v=x$, and if not, which of the connected components of $G-v$ contains $x$. When the input graph is a path, this is equivalent to binary searching in a linearly ordered set. In case of trees, the problem was studied among multiple names including: Tree Search Problem \cite{Cicalese2014ImprovedApproxAvgTs,Cicalese2016OnTSPwNonUniCost}, Search Trees on Trees \cite{SplayTonT,Fast_app_centroid_trees}, Hub Labeling \cite{Angelidakis2018ShortestPQ} and more. For general graphs the problem is known as Vertex Ranking \cite{RankingsofGraphs,Schaffer1989OptNodeRankOfTsInLinTime,OnakParys2006GenOfBSSInTsAndFLikePosets,Mozes_Onak2008FindOptTSStartInLinTime}, Elimination Trees \cite{Pothen1988OptimalEliminationTrees} and Hierarchical Clustering \cite{Acostfunctionforsimilaritybasedhierarchicalclustering,HCObjFsandAlgs,Approximatehierarchicalclusteringviasparsestcutandspreadingmetrics}.

In order to efficiently locate the target, the searcher needs a \textit{strategy} of searching which allows them to locate the target efficiently. The strategy is an \textit{adaptive} algorithm, which given previous responses in finite (and polynomial) time outputs the next query to be performed. We model this strategy as a rooted tree which we will call a \textit{decision tree} $D$, whose nodes are vertices of $G$. We will require that each edge outgoing from the root $r$ of $D$ is associated with a unique response to a query to $r$ in $G$ and that the same holds for all decision subtrees of $D-r$ (in the respective components of $G-v$). When searching according to $D$, the root $r$ of $D$ is queried first. If $r\neq x$ the searcher moves down along the edge of $r$ associated with the response, thereby entering the subtree $D_u$ rooted at the next queried vertex $u$. This process recurses until the target is found.

We will allow each vertex $v\in V\br{G}$ to have an arbitrary \textit{query cost} $c(v)$. 
Intuitively, $c(v)$ represents the number of time units required for a query to $v$ to return an answer.  We will also allow each vertex to have an arbitrary \textit{weight} $w\br{v}$, which denotes how frequently 
the vertex is searched for. For a target vertex $x \in V\br{G}$ and a decision tree $D$, let $Q_G\br{D,x}$ denote the sequence of queries performed 
along the unique path in $D$ from the root $r\br{D}$ to $x$. The \textit{weighted average-case cost} of $D$ on $G$ is defined as:
$$
c_G\br{D} = \sum_{x\in V\br{G}}w\br{x}\cdot \sum_{q\in Q_G\br{D, x}}c\br{q}.
$$

The \textit{Graph Search Problem} (GSP) is to find the decision tree which minimizes the above cost.

\subsection{Motivations and Applications}

Fast retrieval of information in graph structures is a well-studied problem, 
beginning with the seminal work of Knuth \cite{Knuth1973}. 
When the underlying search space provides non-local information about the target, 
the search process can be accelerated, since each query rules out a large set of 
possible target locations. The goal of designing search strategies is to find ways of exploiting this property as effectively as possible. 

Searching arises in various practical problems, but it can also be reformulated, to fit a wide range of real-life applications. Searching in graphs can be used to model a variety of problems, that may initially appear unrelated. These include: scheduling of parallel database join operations \cite{DereniowskiEfPQProcByGRank,OnMinERSTs,MinERSTrofTGs}, automated bug detection in computer code \cite{OptimalSinT,dereniowski2022CFApproxAlgForBSInTsWithMonoQTimes,dereniowski2024SInTsMonoQTs,szyfelbein2025searchingtreeskupmodularweight}, parallel Cholesky factorization of matrices \cite{Dereniowski2003CholeskyFactofMx}, VLSI-layouts \cite{OnAGPartWAppVLSI}, hierarchical clustering of data \cite{Acostfunctionforsimilaritybasedhierarchicalclustering,HCObjFsandAlgs,Approximatehierarchicalclusteringviasparsestcutandspreadingmetrics} and parallel assembly of multi-part products from their components \cite{ParAofModPs,Dereniowski2009ERankOfWTs}.

We focus on the average-case version of the problem rather than the worst-case, 
since it is natural to assume that the search strategies we design are intended 
to be used repeatedly. 
This motivates the introduction of a weight function, as some vertices may serve 
as targets more frequently than others. 
We also allow arbitrary query costs, since performing a query may require 
significant resources, such as time or money. To the best of our knowledge, this particular variant of the problem has not yet been investigated in the literature.

\subsection{Organization of the paper and our results}

In Section \ref{notionsAndPreliminaries}, we introduce all necessary notions, 
preliminaries, and formal definitions required for the analysis, including 
decision trees and their cost (Section \ref{definitionOfDecisionTree}), 
as well as cuts and separators (Section \ref{cutsAndSeparators}).

In Section \ref{serachingInTs}, we consider the case in which the input graph 
is a tree.
Using a pseudoexact FPTAS for the Weighted 
$\alpha$-Separator Problem, in Section \ref{HowToSearchInTs} we present a 
$\br{4+\epsilon}$-approximation algorithm running in $O\br{n^4/\epsilon^2}$ time.

In Section \ref{serachingInGs}, we shift our focus to general graphs. 
Using the $O\br{\sqrt{\log n}}$-approximation algorithm for the Min-Ratio 
Vertex Cut Problem from \cite{Improvedapproximationalgorithmsvertexseparators}, 
we obtain a polynomial-time $O\br{\sqrt{\log n}}$-approximation algorithm 
for the Graph Search Problem.

In Appendix \ref{relatedwork}, we provide an overview of the related work. In Appendix \ref{hardness}, we show that the problem is NP-hard even 
when restricted to trees with $\Delta\br{T}\leq 16$ and to trees with $\diam\br{T}\leq 8$. Appendix \ref{separatorFPTAS} shows how to obtain the aforementioned pseudoexact FPTAS for the Weighted $\alpha$-separator problem and Appendix \ref{omittedproofs} contains other omitted proofs.

\section{Notions and Preliminaries}\label{notionsAndPreliminaries}
We assume that every graph $G$ considered is simple and connected. 
By $uv \in E\br{G}$ we denote an edge connecting vertices $u$ and $v$ in $G$. 
Let $v \in V\br{G}$. By $G-v$ we denote the set of connected components 
resulting from deleting $v$ from $G$. 
For a set $S \subseteq V\br{G}$, $G-S$ denotes the set of 
connected components resulting from deleting all vertices in $S$ from $G$. 
For a family of subsets $\mathcal{F}$ of $V\br{G}$, we define 
$G-\mathcal{F} = G-\bigcup_{S \in \mathcal{F}} S$. The set of neighbors of a vertex $v \in V\br{G}$ is denoted by 
$N_G\br{v} = \brc{u \in V\br{G} \mid uv \in E\br{G}}$, 
and the set of neighbors of a subgraph $\mathcal{G}$ of $G$ is 
$N_G\br{\mathcal{G}} = \bigcup_{v \in V\br{\mathcal{G}}} N_G\br{v} - V\br{\mathcal{G}}$. For any function $f: V\br{G} \to \mathbb{N}$ and any set $S \subseteq V\br{G}$, 
we define $f\br{S} = \sum_{v \in S} f\br{v}$, 
and for a graph $G$, $f\br{G} = f\br{V\br{G}}$.  

Let $T$ be a tree. If $T$ is rooted, its root is denoted by $r\br{T}$. 
For any vertex $v \in V\br{T}$, let 
$\mathcal{C}_{T,v} = \brc{c_1, c_2, \dots, c_{\deg_{T,v}^+}}$ 
be the set of children of $v$. By $T_v$ we will denote the subtree of $T$ rooted at $v$ with all its descendants, and by $T_{v,i}$ we will denote the subtree of $T$ consisting of $v$ and $T_{c_1},\dots, T_{c_i}$.
For a subset $S \subseteq V\br{T}$, $T\angl{S}$ denotes the minimal connected 
subtree of $T$ containing all vertices in $S$.

A \textit{Graph Search Instance} consists of a triple 
$\br{G, c, w}$, where 
$c: V\br{G} \to \mathbb{N}$ is the cost of querying a vertex, and 
$w: V\br{G} \to \mathbb{N}$ is the weight of a vertex. During the \textit{Search Process}, the searcher is allowed to iteratively 
perform \textit{queries}, each asking about a chosen vertex $v$. 
After time $c(v)$, the query returns an answer. 
If the answer is affirmative, then $v$ is the target, otherwise, 
a unique connected component of $G-v$ containing the target $x$ is returned. 
Based on this information, the searcher narrows the subgraph of $G$ that 
might contain $x$. 
We call such subgraph a \textit{candidate} subgraph, and its vertices 
the \textit{candidate} set \footnote{Note that $G$ itself is always a candidate subgraph as well.}. 
This process continues until the position of the target is revealed.

\subsection{Decision trees}\label{definitionOfDecisionTree}

\begin{figure}[htbp]
    \centering
    \begin{minipage}{0.48\textwidth}
        \centering
        \begin{tikzpicture}[thick]
            \node (a) at (0.5,-1.5) {a};
            \node (b) at (0.75,0) {b};
            \node (c) at (1.5,-0.25) {c};
            \node (d) at (-1,-2) {d};
            \node (e) at (1,-3) {e};
            \node (f) at (-2,-1) {f};
            \node (g) at (1,-5.5) {g};
            \node (h) at (2,-2.25) {h};
            \node (i) at (-0.5,-3.5) {i};
            \node (j) at (0.75, -4.25) {j};
            \node (k) at (-1.5,-5) {k};
            \node (l) at (-0.75,-6) {l};
            
            \draw (a) -- (b);
            \draw (a) -- (c);
            \draw (a) -- (d);
            \draw (a) -- (e);
            \draw (b) -- (c);
            \draw (e) -- (j);
            \draw (e) -- (h);
            \draw (d) -- (i);
            \draw (d) -- (f);
            \draw (j) -- (g);
            \draw (i) -- (j);
            \draw (i) -- (e);
            \draw (i) -- (k);
            \draw (g) -- (l);
            \draw (h) -- (a);
            \draw (l) -- (k);
        \end{tikzpicture}
        \label{fig:graph}
    \end{minipage}
    \hfill
    \begin{minipage}{0.48\textwidth}
        \centering
        \begin{tikzpicture}[thick]
            \node (e) at (0,0) {$e$};
            \node (i) at (0,-1.5) {$i$};
            \node (a) at (-1.5,-3) {$a$};
            \node (g) at (1.5,-3) {$g$};
            \node (c) at (-2.5,-4.5) {$c$};
            \node (h) at (-1.5,-4.5) {$h$};
            \node (f) at (-0.5,-4.5) {$f$};
            \node (j) at (1,-4.5) {$j$};
            \node (k) at (2,-4.5) {$k$};
            \node (b) at (-3,-6) {$b$};
            \node (d) at (-0.5,-6) {$d$};
            \node (l) at (2,-6) {$l$};

            \draw[->] (e) -> (i);
            \draw[->] (i) -- (a);
            \draw[->] (i) -- (g);
            \draw[->] (a) -- (c);
            \draw[->] (a) -- (h);
            \draw[->] (a) -- (f);
            \draw[->] (c) -- (b);
            \draw[->] (f) -- (d);
            \draw[->] (g) -- (j);
            \draw[->] (g) -- (k);
            \draw[->] (k) -- (l);
        \end{tikzpicture}
    \end{minipage}
    \caption[Graph and decision trees for it]{Sample input graph (on the left) and a decision tree for it (on the right).}
    \label{fig:sample_decision_trees_for_graph}
\end{figure}
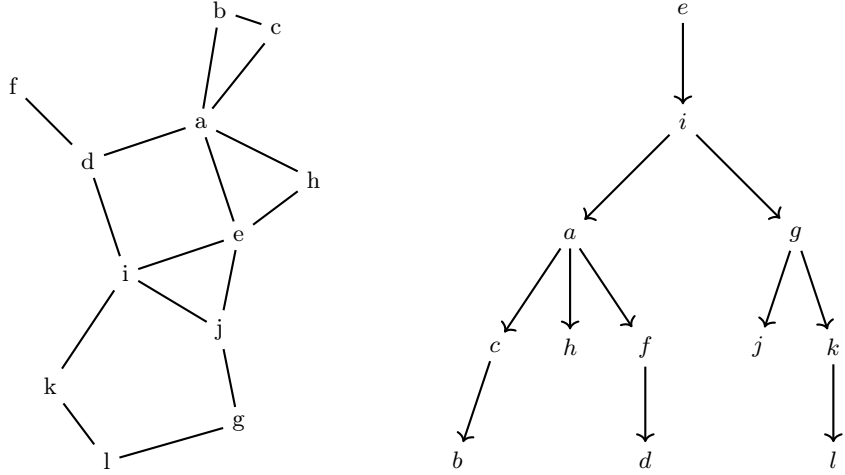

A decision tree is a pair $D = \br{V\br{D}, E\br{D}}$, where $V\br{D}=V\br{G}$ are vertices and $E\br{D}$ are edges of $D$. It is required that each child of $q\in V\br{
D}$ corresponds to a distinct response to the query at $q$, with respect to the subtree of candidate solutions remaining 
after performing all previous queries.

Let $Q_G\br{D,x}$ denote the sequence of queries made to locate $x \in V\br{G}$ using decision tree $D$. 
We define the \textit{Graph Search Problem} as follows:
\begin{tcolorbox}[colback=white, title=Graph Search Problem (GSP), fonttitle=\bfseries, breakable]
\textbf{Input:} Graph $G$, a query cost function $c:V\to \mathbb{N}$ and a weight function $w:V\to \mathbb{N}$.

\textbf{Output:} A decision tree $D$, minimizing the weighted average search cost: 
$$
c_G\br{D} = \sum_{x\in V\br{G}}w\br{x}\cdot \sum_{q\in Q\br{D, x}}c\br{q}.
$$
\end{tcolorbox}

If a given decision tree is optimal, we denote its cost by $\OPT\br{G}$. 
Let $D'$ be a subtree of a valid decision tree $D$ for $T$ containing $r\br{D}$. 
We say that $D'$ is a \textit{partial} decision tree for $D$, 
and we define its cost analogously to that of $D$, 
although we only count the queries belonging to $D'$.
We also introduce the following reinterpretation of the cost function, 
for each node $v \in D$, let $G_{D,v}$ be the subgraph of $G$ in which 
$v$ is queried when using $D$. 
Then, the contribution of $v$ to the total cost is $w\br{G_{D,v}} \cdot c\br{v}$, 
and therefore we obtain the following simple lemma:

\begin{lemma}\label{contributionLemma}
$$
c_G\br{D} = \sum_{v \in V\br{G}} w\br{G_{D,v}} \cdot c\br{v}.
$$  
\end{lemma}


\subsection{Cuts and separators}\label{cutsAndSeparators}
To obtain a tight lower bound on the cost of our solution, 
we establish a connection between the Graph Search Problem and the following 
vertex separator problems. 
We define the \textit{Weighted $\alpha$-Separator Problem} as follows:

\begin{tcolorbox}[colback=white, title=Weighted $\alpha$-Separator Problem, fonttitle=\bfseries, breakable]
\textbf{Input:} Graph $G$, a cost function $c:V\to \mathbb{N}$, a weight function $w:V\to \mathbb{N}$ and a real number $\alpha$.

\textbf{Output:} A set $S\subseteq V\br{G}$ called \textit{$\alpha$-separator}, such that for every $H\in G-S$, $w\br{H}\leq w\br{G}/\alpha$ and $c\br{S}$ is minimized.
\end{tcolorbox}
We also define the Min-Ratio Vertex Cut Problem as follows:

\begin{tcolorbox}[colback=white, title= Min-Ratio Vertex Cut Problem, fonttitle=\bfseries, breakable]
\textbf{Input:} Graph $G=\br{V\br{G}, E\br{G}}$, the cost function $c:V\to \mathbb{N}$ and the weight function $w:V\to \mathbb{N}$.

\textbf{Output:} A partition $\br{A,S,B}$ of $V\br{G}$ called \textit{vertex-cut}, such that there are no $u\in A$ and $v\in B$ for which $uv\in E\br{G}$, minimizing the ratio:
$$
\alpha_{c,w}\br{A,S,B}=\frac{c\br{S}}{w\br{A\cup S}\cdot w\br{B\cup S}}.
$$
\end{tcolorbox}
\subsection{Levels of $\OPT$ and basic bounds}\label{levelsOfOptAndBasicBounds}

We begin with additional notation. 
For any graph $G$ and decision tree $D$, denote by 
$\mathcal{R}_D\br{G} = \brc{V\br{G_{D,v}} | v \in V\br{G}}$ 
the family of all candidate subsets of $D$ in $G$. 

Let $D^*$ be an arbitrary decision tree for the Graph Search Problem such that 
$c_G\br{D^*} = \OPT\br{G}$. 
We denote by $\mathcal{L}_{k}^*$ the subfamily of $\mathcal{R}_{D^*}\br{G}$ 
consisting of all maximal elements $H$ of $\mathcal{R}_{D^*}\br{G}$ with $w\br{H} \leq k$, 
that is, if some superset $H'$ of $H$ belongs to $\mathcal{R}_{D^*}\br{G}$, 
then $w\br{H'} > k$. 
We call such a set the $k$-th \textit{level} of $\OPT\br{G}$. 
Let $S_{k}^* = V\br{G} - \mathcal{L}_{k}^*$. 
These are the vertices belonging to the separator at level $\mathcal{L}_{k}^*$. 

Notice that $S_{k}^*$ forms a Weighted $w\br{G}/k$-separator of $G$. 
Furthermore, for any $H_1, H_2 \in \mathcal{R}_D\br{G}$, we have 
$H_1 \cup H_2 \neq \emptyset$ if and only if $H_1 \subseteq H_2$ or 
$H_2 \subseteq H_1$, so $\mathcal{R}_D\br{G}$ is laminar. 
Therefore, for any $k_1 \neq k_2$, we have 
$\mathcal{L}_{k_1}^* \cap \mathcal{L}_{k_2}^* = \emptyset$.
\begin{lemma}
                $$\OPT\br{G}=\sum_{k=0}^{w\br{G}-1}c\br{S_{k}^*}.$$
            \end{lemma}
            \begin{proof}
                Consider any vertex $v$. For every $0\leq k<w\br{G_{D^*,v}}$, $v\notin \bigcup_{H\in \mathcal{L}_{k}^*}H$, so $v\in S_{k}^*$ and the contribution of $v$ to the cost is $w\br{G_{D^*,v}}\cdot c\br{v}$:
                $$\sum_{k=0}^{w\br{G}-1}c\br{S_{k}^*}=\sum_{v\in V\br{G}}\sum_{k=0}^{w\br{G_{D^*,v}}-1}c\br{v}=\OPT\br{G}$$
                
                where the second equality is by Lemma \ref{contributionLemma}.
            \end{proof}
            
Using the above lemma one easily obtains the following lower bound on the cost of the optimal solution:
\begin{lemma}\label{lb_opt}
            $$
            2\cdot\OPT\br{G}= 2\cdot\sum_{k=0}^{w\br{T}-1}c\br{S_{k}^*} \geq \sum_{k=0}^{w\br{T}}c\br{S_{\fl{k/2}}^*}.
            $$
\end{lemma}

We also have the following upper bound:
\begin{lemma}\label{splitting}
    Let $\mathcal{G}$ be any subgraph of $G$ and $0\leq\beta\leq 1$. Then: 
            $$
           \beta\cdot w\br{\mathcal{G}}\cdot c\br{S_{\fl{w\br{\mathcal{G}}/2}}^*\cap \mathcal{G}}
            \leq \sum_{k=\br{1-\beta}w\br{\mathcal{G}}+1}^{w\br{\mathcal{G}}}c\br{S_{\fl{k/2}}^*\cap \mathcal{G}}.
            $$
\end{lemma}

    \begin{proof}
        The inequality is due to the fact that as $k$ decreases, more vertices belong to the separator. 
    \end{proof}
    
We will also make use of the following lemma, which allows us to attach decision trees 
below already constructed partial decision trees without ambiguity:
\begin{lemma}\label{neighborsPathLemma}
     Let $\mathcal{G}$ be a connected subgraph of a graph $G$ and $D$ be a partial decision tree for $G$ having no queries to vertices of $\mathcal{G}$, but containing at least one query to $N_{G}\br{V\br{\mathcal{G}}}$. Additionally, let $Q$ be the set of all queries to vertices from $N_{G}\br{V\br{\mathcal{G}}}$ in $D$. Then $D\angl{Q}$ forms a path in $D$. 
\end{lemma}
    \begin{proof}
        See Appendix \ref{proof_neighborsPathLemma}.
    \end{proof}
    
The above lemma will become useful in the following scenario: 
Let $D$ be a partial decision tree for $G$ satisfying the conditions of the lemma and $D_{\mathcal{G}}$ be any partial decision tree for a subgraph $\mathcal{G}$ of $G$. 
Since the queries in $Q$ form a path, without ambiguity, we can attach $D_{\mathcal{G}}$ to $D$ 
below the last query of $Q$, thereby obtaining a valid partial decision tree for $G$.

\section{Searching in Trees}\label{serachingInTs}

In this section, we present a $\br{4+\epsilon}$-approximation algorithm 
for the case where the input graph is a tree. 
To achieve this, we establish a connection between searching in trees and 
the Weighted $\alpha$-Separator Problem. 
This connection provides a lower-bounding scheme for our recursive algorithm, 
which at each level of recursion, constructs a decision tree using the 
$\alpha$-separator obtained by the following procedure:

\begin{theorem}\label{bicriteriaFPTAS}
    Let $S$ be an optimal weighted $\alpha$-separator for $\br{T,c,w,\alpha}$. For any $\delta>0$ there exists an algorithm \FSeparatorFPTAS, which returns a separator $S'$, such that:
    \begin{enumerate}
        \item $c\br{S'}\leq c\br{S}$.
        \item $w\br{H}\leq \frac{\br{1+\delta}\cdot w\br{T}}{\alpha}$ for every $H\in T-S'$.
        \item The algorithm runs in $O\br{n^3/\delta^2}$ time.
    \end{enumerate}
\end{theorem}
\begin{proof}
    The algorithm and the proof are deferred to the Appendix \ref{separatorFPTAS}.
\end{proof}

\subsection{How to search in trees}\label{HowToSearchInTs}
Below, we show how to use the $\FSeparatorFPTAS$ procedure to construct a solution 
for the Tree Search Problem. 
At each level of the recursion, the algorithm greedily finds an (almost) optimal 
weighted $\alpha$-separator of $T$, denoted $S_T$, and then builds an arbitrary 
decision tree $D_T$ using the vertices in $S_T$ (which can be done in $O\br{n^2}$ time). 

Next, for each $H \in T-S_T$, the procedure is called recursively, and each resulting 
decision tree $D_H$ is attached below the appropriate query in $D_T$. The resulting decision tree is then returned by the procedure.
    \begin{theorem}
        For any $\epsilon>0$, there exists $\br{4+\epsilon}$-approximation algorithm for the Tree Search Problem running in time $O\br{n^4/\epsilon^2}$.
    \end{theorem}
        \begin{proof}
            The procedure is as follows: 
            
\begin{algorithm}[H]
\caption{The $\br{4+\epsilon}$-approximation algorithm for the Tree Search Problem}
\label{createDecisionTree}
\SetKwFunction{FDecisionTree}{DecisionTree}
\SetKwFunction{FSeparatorFPTAS}{SeparatorFPTAS}
\SetKwProg{Fn}{Procedure}{:}{}
\Fn{$\FDecisionTree\br{T, c, w,  \epsilon}$}{
$S_T\gets\FSeparatorFPTAS\br{T, c, w, \alpha=2, \delta = \frac{\epsilon}{4+\epsilon}}$.

$D_T\gets$ arbitrary partial decision tree for $T$, built from vertices of $S_T$.

    \ForEach{$H\in T-S_T$}
    {
        $D_H\gets \FDecisionTree\br{H, c, w, \epsilon}$.

        Hang $D_H$ in $D_T$ below the last query to $v\in N_T\br{H}$.
    }   
    \Return $D_T$.
    
}
\end{algorithm}
            \begin{figure}[htbp]
    \centering
    \begin{minipage}{0.48\textwidth}
        \centering
        \begin{tikzpicture}[scale=0.7]
    \draw[thick, fill=white, drop shadow]
  (0,0) 
  .. controls (1,0) and (1,-6) .. (0,-6)  
  .. controls (-1,-6) and (-1,0) .. (0,0); 
  
  \node at (-0.75, 0)  {$S_T$};

Dots inside
\foreach \y in {-1,-2,-3,-5} {
  \fill (0,\y) circle (4pt);
}
\node at (0, -3.9) {$\vdots$};

\draw[thick, fill=white, drop shadow]
(3,0.75) 
  .. controls (4.5,0.75) and (4.5,-0.75) .. (3,-0.75)  
  .. controls (1.5,-0.75) and (1.5,0.75) .. (3,0.75); 

\node at (3, -0.05)  {$H_1$};

\draw[thick] (0.45,-0.5) -- (1.87,0);

\draw[thick] (0.745,-2.5) -- (2.04,-0.4);

\draw[thick, fill=white, drop shadow]
(3.5,-1) 
  .. controls (4.5,-1) and (4.5,-2) .. (3.5,-2)  
  .. controls (2.5,-2) and (2.5,-1) .. (3.5,-1); 

\node at (3.5, -1.55)  {$H_2$};

\draw[thick] (0.66,-1.6) -- (2.74,-1.5);

\draw[thick, fill=white, drop shadow]
(3.25,-2.25) 
  .. controls (4.5,-2.25) and (4.5,-3.5) .. (3.25,-3.5)  
  .. controls (2,-3.5) and (2,-2.25) .. (3.25,-2.25); 

\node at (3.25, -2.9)  {$H_3$};

\draw[thick] (0.75,-2.6) -- (2.32,-2.8);

\draw[thick] (0.75,-3.5) -- (2.32,-3);

\draw[thick] (0.58,-5) -- (2.6,-3.35);

\node at (3.15, -3.95) {$\vdots$};

\draw[thick, fill=white, drop shadow]
(3,-4.65) 
  .. controls (4.75,-4.65) and (4.75,-6.5) .. (3,-6.5)  
  .. controls (1.25,-6.5) and (1.25,-4.65) .. (3,-4.65); 
  
\node at (3, -5.6)  {$H_p$};

\draw[thick] (0.69,-4.25) -- (1.85,-5.1);

\draw[thick] (0.45,-5.5) -- (1.7,-5.75);

\end{tikzpicture}
    \end{minipage}
\begin{minipage}
    {0.48\textwidth}
    \centering
    \begin{tikzpicture}[scale=1]
            \draw[thick, fill=gray!30, drop shadow] (4,-4) -- (4.9,-5.8) -- (3.1,-5.8) -- cycle
                  node[right] {$D_{T}$};
                  
            \draw[thick, fill=white, drop shadow] (2.5,-6.5) -- (3.2,-7.9) -- (1.8,-7.9) -- cycle
            node[right] {$D_{H_1}$};

            \draw[thick, fill=white, drop shadow] (3.8,-6.4) -- (4.3,-7.4) -- (3.3,-7.4) -- cycle
            node[right] {$D_{H_2}$};
            
            \draw[thick, fill=white, drop shadow] (4.8,-6.3) -- (5.2,-7.1) -- (4.4,-7.1) -- cycle
            node[right] {$D_{H_3}$};
            
            \draw[thick, fill=white, drop shadow] (6.5,-6.5) -- (7.2,-8.1) -- (5.8,-8.1) -- cycle
            node[right] {$D_{H_p}$};
            
            \node at (5.75, -7) {$\dots$};
            
            \draw[thick] (3.1,-5.8) -- (2.5,-6.5);
            \draw[thick] (3.6,-5.8) -- (3.8,-6.4);
            \draw[thick] (4.2,-5.8) -- (4.8,-6.3);
            
            \draw[thick] (4.9,-5.8) -- (6.5,-6.5);
            
        \end{tikzpicture}
\end{minipage}
    \caption{The separator $S_T$ produced by the algorithm and the structure of the decision tree built using $S_T$.}
    \label{fig:placeholder}
\end{figure}
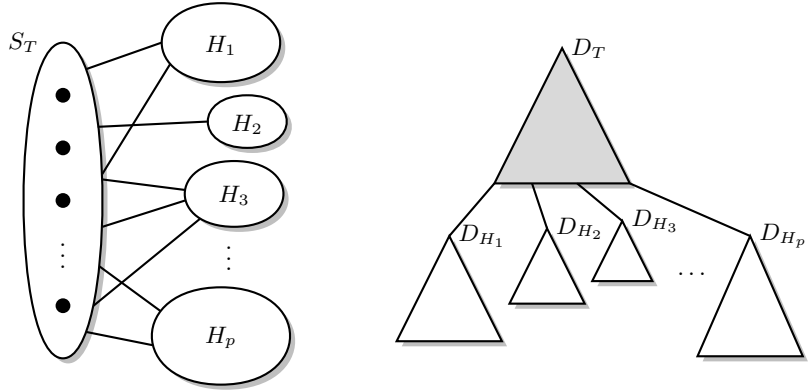
        
            Let $\mathcal{T}$ be a subtree of $T$ for which the procedure was called and let $S_{\mathcal{T}}^*=S_{\fl{w\br{\mathcal{T}}/2}}^*\cap\mathcal{T}$. By Theorem \ref{bicriteriaFPTAS}, we have that $c\br{S_{\mathcal{T}}}\leq c\br{S_{\mathcal{T}}^*}$. Using $\beta=\frac{1-\delta}{2}$ and applying Lemma \ref{splitting} we have that the contribution of the decision tree $D_{\mathcal{T}}$ is bounded by:
            \begin{align*}
                w\br{\mathcal{T}}\cdot c\br{S_{\mathcal{T}}}
                &\leq w\br{\mathcal{T}}\cdot c\br{S_{\mathcal{T}}^*}\leq \frac{2}{1-\delta}\cdot \sum_{k=\frac{1+\delta}{2}\cdot w\br{\mathcal{T}}+1}^{w\br{\mathcal{T}}}c\br{S_{\fl{k/2}}^*\cap \mathcal{T}}.
            \end{align*}
            
            To bound the cost of the whole solution we will firstly show the following lemma which is necessary to proceed:
            \begin{lemma}\label{up_trees}
            $$\sum_{\mathcal{T}}\sum_{k=\frac{1+\delta}{2}\cdot w\br{\mathcal{T}}+1}^{w\br{\mathcal{T}}}c\br{S_{\fl{k/2}}^*\cap \mathcal{T}}\leq \sum_{k=0}^{w\br{T}}c\br{S_{\fl{k/2}}^*}.$$
            \end{lemma}
            \begin{proof}
                Fix a value of $\mathcal{T}$ and $k$. Their contribution to the cost is $c\br{S_{\fl{k/2}}^*\cap \mathcal{T}}$. Consider which candidate subtrees contribute such a term. As $S_{\mathcal{T}}$ is a weighted $\frac{2}{1+\delta}$-separator, we have that $\mathcal{T}$ is the minimal candidate subtree, such that $w\br{\mathcal{T}}\geq k\geq \frac{\br{1+\delta}\cdot w\br{\mathcal{T}}}{2}+1 >w\br{H}$, for every $H\in \mathcal{T}-S_{\mathcal{T}}$. This means that if for every $H\in \mathcal{T}-S_{\mathcal{T}}$, $w\br{H}<k$, then $\mathcal{T}$ contributes such a term. Since for all $H_1, H_2\in \mathcal{T}-S_{\mathcal{T}}$ we have that $H_1\cap H_2=\emptyset$, $\br{S_{\fl{k/2}}^*\cap H_1}\cup \br{S_{\fl{k/2}}^*\cap H_2}=\emptyset$, the claim follows by summing over all values of $k$.
            \end{proof}
            
        We are now ready to bound the cost of the solution. Let $D$ be the decision tree returned by the procedure. Using the fact that by definition $\frac{4}{1-\delta}=4+\epsilon$, we have:
        \begin{align*}
            c_T\br{D} &\leq \sum_{\mathcal{T}} w\br{\mathcal{T}}\cdot c\br{S_\mathcal{T}}
            \leq \frac{2}{1-\delta}\cdot\sum_{\mathcal{T}}\sum_{k=\frac{1+\delta}{2}\cdot w\br{\mathcal{T}}+1}^{w\br{\mathcal{T}}}c\br{S_{\fl{k/2}}^*\cap \mathcal{T}}\\
            &\leq \frac{2}{1-\delta}\cdot\sum_{k=0}^{w\br{T}}c\br{S_{\fl{k/2}}^*}\leq \frac{4}{1-\delta}\cdot\OPT\br{T} = \br{4+\epsilon}\cdot\OPT\br{T}
    \end{align*}

    where the third inequality is due to Lemma \ref{up_trees} and the last inequality is by Lemma \ref{lb_opt}.

    As $1/\delta=\frac{4+\epsilon}{\epsilon}=1+4/\epsilon$ and each $v\in V\br T$ belongs to the set $S_{\mathcal{T}}$ exactly once, we have that the overall running time is at most $O\br{n^4/\epsilon^2}$ as required.
        \end{proof}

\section{Searching in general graphs}\label{serachingInGs}
To construct decision trees for general graphs, we exploit a connection to a different problem, 
namely the Min-Ratio Vertex Cut Problem. 
Let $\alpha_{c,w}\br{G}$ denote the optimal value of such a vertex cut. 
We invoke the following result given in \cite{Improvedapproximationalgorithmsvertexseparators}:

\begin{theorem}\label{approxmrvc}
    Given a graph $G=\br{V\br{G}, E\br{G}}$, the cost function $c:V\to\mathbb{N}$ and the weight function $w:V\to \mathbb{N}$, there exists a
polynomial-time algorithm, which computes a partition $(A, S, B)$, such that:
$$
\alpha_{c,w}\br{A,S,B}=O\br{\sqrt{\log n
}}\cdot\alpha_{c,w}\br{G}.
$$
\end{theorem}

Combining the latter procedure, with the following result, yields an $O\br{\sqrt{\log n
}}$-approximation algorithm for the Graph Search Problem:
\begin{theorem}
    Let $f_n$ be the approximation ratio of any polynomial time algorithm for the Min-Ratio Vertex Cut Problem. Then, there exists an $O\br{f_n}$-approximation algorithm for the GSP, running in polynomial time.
\end{theorem}
    \begin{proof}
        Let $\FAlgorithmMinCut$ be the procedure which achieves the $f_n$-approximation ratio for the Min-Ratio Vertex Cut Problem.
        The algorithm is as follows:
        
\begin{algorithm}[H]
\caption{The $f_n$-approximation algorithm for the Graph Search Problem.}
\label{createDecisionTree}
\SetKwFunction{FDecisionTree}{DecisionTree}
\SetKwProg{Fn}{Procedure}{:}{}
\Fn{$\FDecisionTree\br{G,c,w}$}{
$A_G,S_G, B_G\gets\FAlgorithmMinCut\br{G, c, w}$.

$D_G\gets$ arbitrary partial decision tree for $G$, built from vertices of $S_G$.

    \ForEach{$H\in G-S_G$}
    {
        $D_H\gets \FDecisionTree\br{H, c, w}$.

        Hang $D_H$ in $D_G$ below the last query to $v\in N_G\br{H}$.
    }   
    \Return $D_G$.
    
}
\end{algorithm}
        
        Let $\mathcal{G}$ be any subgraph of $G$, for which the procedure was called and let $S_{\mathcal{G}}^*=S_{\fl{w\br{\mathcal{G}}/2}}^*\cap \mathcal{G}$.
        \begin{lemma}\label{lambda_lemma}
            Let $\mathcal{H}=\mathcal{G}-S_{\mathcal{G}}^*$ and let $\lambda=6+2\sqrt{5}$ be the unique, positive solution of the equation $\frac{1}{4}-\frac{1}{2\sqrt{\lambda}}=\frac{1}{\lambda}$. Then, we can partition $\mathcal{H}$ into two sets, $\mathcal{A}$ and $\mathcal{B}$ such that for $A=\bigcup_{H\in\mathcal{A}}V\br{H}$ and $B=\bigcup_{H\in\mathcal{B}}V\br{H}$, we have:
            $$w\br{A\cup S_{\mathcal{G}}^*}\cdot w\br{B\cup S_{\mathcal{G}}^*}\geq w\br{\mathcal{G}}^2/\lambda.$$
        \end{lemma}

            \begin{proof}
                See Appendix \ref{proof_lambda_lemma}.
            \end{proof}
        
        Using the fact, that the partition $\br{A,S_{\mathcal{G}}^*,B}$ in the above lemma is a vertex cut of $\mathcal{G}$, we have the following upper bound on the optimal value of the min-ratio-vertex cut of $\mathcal{G}$,  $\alpha_{c,w}\br{\mathcal{G}}$:
                $$\alpha_{c,w}\br{\mathcal{G}}\leq \alpha_{c,w}\br{A,S_{\mathcal{G}}^*,B}=\frac{c\br{S_{\mathcal{G}}^*}}{w\br{A\cup S_{\mathcal{G}}^*}\cdot w\br{B\cup S_{\mathcal{G}}^*}}\leq \frac{\lambda\cdot c\br{S_{\mathcal{G}}^*}}{w\br{\mathcal{G}}^2}.
                $$
                
        Let $\br{A_{\mathcal{G}},S_{\mathcal{G}}, B_{\mathcal{G}}}$, be the partition returned by $\FAlgorithmMinCut \br{\mathcal{G},c,w}$. Without the loss of generality assume that $w\br{A_{\mathcal{G}}}\geq w\br{B_{\mathcal{G}}}$. Using Theorem \ref{approxmrvc}, we get that:
        $$
        \alpha_{c,w}\br{A_{\mathcal{G}},S_{\mathcal{G}}, B_{\mathcal{G}}}=\frac{c\br{S_{\mathcal{G}}}}{w\br{A_{\mathcal{G}}\cup S_{\mathcal{G}}}\cdot w\br{B_{\mathcal{G}}\cup S_{\mathcal{G}}}}\leq f_n \cdot \frac{\lambda\cdot c\br{S_{\mathcal{G}}^*}} {w\br{\mathcal{G}}^2}.
        $$
        
        Let $\beta=w\br{B_{\mathcal{G}}\cup S_{\mathcal{G}}}/w\br{\mathcal{G}}$. We have that $\br{1-\beta}\cdot w\br{\mathcal{G}} = w\br{A_\mathcal{G}}$, so we conclude that the contribution of the decision tree $D_{\mathcal{G}}$ is bounded by:
        \begin{align*}
        w\br{\mathcal{G}}\cdot c\br{S_{\mathcal{G}}} &\leq \lambda \cdot f_n\cdot \frac{w\br{A_{\mathcal{G}}\cup S_{\mathcal{G}}}\cdot w\br{B_{\mathcal{G}}\cup S_{\mathcal{G}}}}{w\br{\mathcal{G}}}\cdot c\br{S_{\mathcal{G}}^*}\\
        & \leq 
        \lambda \cdot f_n\cdot w\br{B_{\mathcal{G}}\cup S_{\mathcal{G}}}\cdot c\br{S_{\mathcal{G}}^*} \leq 
        \lambda \cdot f_n\cdot \sum_{k=w\br{A_{\mathcal{G}}}+1}^{w\br{\mathcal{G}}}c\br{S_{\fl{k/2}}^*\cap \mathcal{G}}
        \end{align*}

        where the last inequality is by Lemma \ref{splitting}.

        As before, to bound the cost of the whole solution we will firstly show the following lemma. The argument is mostly the same as for the Lemma \ref{up_trees}, however, there are few differences and we include it for completeness:
            \begin{lemma}\label{up_graphs}
            $$\sum_{\mathcal{G}}\sum_{k=w\br{A_{\mathcal{G}}}+1}^{w\br{\mathcal{G}}}c\br{S_{\fl{k/2}}^*\cap \mathcal{G}}\leq \sum_{k=0}^{w\br{G}}c\br{S_{\fl{k/2}}^*}.$$
            \end{lemma}
            \begin{proof}
                See Appendix \ref{proof_of_up_graphs}.
            \end{proof}
            
        We are now ready to bound the cost of the solution. Let $D$ be the decision tree returned by the procedure. We have:
        \begin{align*}
            c_G\br{D}&\leq \sum_{\mathcal{G}}w\br{\mathcal{G}}\cdot c\br{S_{\mathcal{G}}}\\
            &\leq \lambda \cdot f_n\cdot \sum_{\mathcal{G}}\sum_{k=w\br{A_{\mathcal{G}}}+1}^{w\br{\mathcal{G}}}c\br{S_{\fl{k/2}}^*\cap \mathcal{G}}\leq 
            \lambda \cdot f_n\cdot\sum_{k=0}^{w\br{G}}c\br{S_{\fl{k/2}}^*} \\
            &\leq 2\cdot \lambda \cdot f_n\cdot \OPT\br{G} = \br{12+4\sqrt{5}}\cdot f_n \cdot \OPT\br{G}
        \end{align*}
        
    where the third inequality is due to Lemma \ref{up_graphs} and the last inequality is by Lemma \ref{lb_opt}.
    \end{proof}

\section{Conclusions and further work}
We have presented an $O\br{\sqrt{\log n}}$-approximation algorithm for the general 
Graph Search Problem, and a $\br{4+\epsilon}$-approximation algorithm for the case 
when the input graph is a tree. To the best of our knowledge, these are the first approximation results for the vertex query and non-uniform costs variant of the problem.

Moreover, the $\br{4+\epsilon}$-approximation algorithm for trees and its analysis can easily be 
adapted to obtain a $\br{4+\epsilon}$-approximation for the edge-query version of the problem, 
thus improving the $6.75$-approximation given in 
\cite{Acostfunctionforsimilaritybasedhierarchicalclustering,HCObjFsandAlgs,Approximatehierarchicalclusteringviasparsestcutandspreadingmetrics}.

The advantage of our solution lies in its simplicity, making it suitable for practical applications. 
This also places the average-case version of the problem in stark contrast to the worst-case variant, 
for which achieving even the state-of-the-art $O\br{\sqrt{\log n}}$-approximation for trees 
required a non-trivial dynamic programming procedure and an intricate analysis 
\cite{dereniowski2017ApproxSsForGeneralBSinWTs}. 
An open question is whether a constant-factor approximation exists for the latter as well.

The current state-of-the-art algorithm for the Min-Ratio Vertex Cut Problem, which we use as an essential 
subroutine for our GSP algorithm, is based on an SDP relaxation. 
An interesting direction for future research is to formulate a similar relaxation for GSP itself 
and to prove that it has an integrality gap of at most $O\br{\sqrt{\log n}}$. 
Another question is whether the approximation ratio for GSP is inherently limited by the 
approximation ratio of Min-Ratio Vertex Cut Problem, or if there exists a method to improve the 
approximation without relying on vertex cuts. 

On the hardness side, we have shown that the problem is NP-hard even when restricted to trees with $\Delta\br{T}\leq 16$ and to trees with $\diam\br{T}\leq 8$. Note, that for general graphs, it can be argued that the problem cannot 
be approximated within any constant factor unless the Small Set Expansion Hypothesis fails, 
by reusing the argument for the edge-query variant provided in 
\cite{Approximatehierarchicalclusteringviasparsestcutandspreadingmetrics}. 
An interesting open question is whether one can establish an inapproximability result 
for some constant $c \leq 4$ in the case when the input graph is a tree.

\begin{credits}

\subsubsection{\discintname}
The authors have no competing interests to declare that are
relevant to the content of this article.
\end{credits}
%
%
%
%

\bibliographystyle{splncs04}
\bibliography{references}

@article{Deligkas2019BsInGsRev,
  author    = {Argyrios Deligkas and George B. Mertzios and Paul G. Spirakis},
  title     = {Binary Search in Graphs Revisited},
  journal   = {Algorithmica},
  volume    = {81},
  number    = {5},
  pages     = {1757--1780},
  year      = {2019},
  month     = {May},
  issn      = {1432-0541},
  doi       = {10.1007/s00453-018-0501-y}
}

@inproceedings{Emamjomeh2016DetAndProbBSinGs,
author = {Emamjomeh-Zadeh, Ehsan and Kempe, David and Singhal, Vikrant},
year = {2016},
month = {06},
pages = {519-532},
title = {Deterministic and probabilistic binary search in graphs},
doi = {10.1145/2897518.2897656}
}

@InProceedings{Jacobs2010OnTheComplexSearchInTsAvg,
    author    = {Tobias Jacobs and Ferdinando Cicalese and Eduardo Laber and Marco Molinaro},
    editor    = {Samson Abramsky and Cyril Gavoille and Claude Kirchner and Friedhelm Meyer auf der Heide and Paul G. Spirakis},
    title     = {On the Complexity of Searching in Trees: Average-Case Minimization},
    booktitle = {Automata, Languages and Programming},
    year      = {2010},
    publisher = {Springer Berlin Heidelberg},
    address   = {Berlin, Heidelberg},
    pages     = {527--539},
    isbn      = {978-3-642-14165-2}
}

@article{Cicalese2014ImprovedApproxAvgTs,
author = {Cicalese, Ferdinando and Jacobs, Tobias and Laber, Eduardo and Molinaro, Marco},
year = {2014},
month = {04},
pages = {},
title = {Improved Approximation Algorithms for the Average-Case Tree Searching Problem},
volume = {68},
journal = {Algorithmica},
doi = {10.1007/s00453-012-9715-6}
}

@article{Hgemo2024TightAB,
  title={Tight Approximation Bounds on a Simple Algorithm for Minimum Average Search Time in Trees},
  author={Svein H{\o}gemo},
  journal={ArXiv},
  year={2024},
  volume={abs/2402.05560},
  url={https://api.semanticscholar.org/CorpusID:267547530}
}

@inbook{SplayTonT,
author = {Berendsohn, Benjamin and Kozma, László},
year = {2022},
month = {01},
pages = {1875-1900},
title = {Splay trees on trees},
isbn = {978-1-61197-707-3},
doi = {10.1137/1.9781611977073.75}
}

@unknown{Fast_app_centroid_trees,
author = {Berendsohn, Benjamin and Golinsky, Ishay and Kaplan, Haim and Kozma, László},
year = {2022},
month = {09},
pages = {},
title = {Fast approximation of search trees on trees with centroid trees},
doi = {10.48550/arXiv.2209.08024}
}

@article{Dereniowski2009ERankOfWTs,
title = {Edge ranking of weighted trees},
journal = {Discrete Applied Mathematics},
volume = {154},
number = {8},
pages = {1198-1209},
year = {2006},
issn = {0166-218X},
doi = {https://doi.org/10.1016/j.dam.2005.11.005},
author = {Dariusz Dereniowski}
}

@InProceedings{dereniowski2017ApproxSsForGeneralBSinWTs,
  author =	{Dereniowski, Dariusz and Kosowski, Adrian and Uznanski, Przemyslaw and Zou, Mengchuan},
  title =	{{Approximation Strategies for Generalized Binary Search in Weighted Trees}},
  booktitle =	{44th International Colloquium on Automata, Languages, and Programming (ICALP 2017)},
  pages =	{84:1--84:14},
  series =	{Leibniz International Proceedings in Informatics (LIPIcs)},
  ISBN =	{978-3-95977-041-5},
  ISSN =	{1868-8969},
  year =	{2017},
  volume =	{80},
  editor =	{Chatzigiannakis, Ioannis and Indyk, Piotr and Kuhn, Fabian and Muscholl, Anca},
  publisher =	{Schloss Dagstuhl -- Leibniz-Zentrum f{\"u}r Informatik},
  address =	{Dagstuhl, Germany},
  URN =		{urn:nbn:de:0030-drops-74507},
  doi =		{10.4230/LIPIcs.ICALP.2017.84}
}

@article{Cicalese2012BinIdentPForWTs,
title = {The binary identification problem for weighted trees},
journal = {Theoretical Computer Science},
volume = {459},
pages = {100-112},
year = {2012},
issn = {0304-3975},
doi = {https://doi.org/10.1016/j.tcs.2012.06.023},
author = {Ferdinando Cicalese and Tobias Jacobs and Eduardo Laber and Caio Valentim}
}

@article{Cicalese2016OnTSPwNonUniCost,
title = {On the tree search problem with non-uniform costs},
journal = {Theoretical Computer Science},
volume = {647},
pages = {22-32},
year = {2016},
issn = {0304-3975},
doi = {https://doi.org/10.1016/j.tcs.2016.07.019},
author = {Ferdinando Cicalese and Balázs Keszegh and Bernard Lidický and Dömötör Pálvölgyi and Tomáš Valla}
}

@article{Schaffer1989OptNodeRankOfTsInLinTime,
title = {Optimal node ranking of trees in linear time},
journal = {Information Processing Letters},
volume = {33},
number = {2},
pages = {91-96},
year = {1989},
issn = {0020-0190},
doi = {https://doi.org/10.1016/0020-0190(89)90161-0},
author = {Alejandro A. Schäffer}
}

@INPROCEEDINGS{OnakParys2006GenOfBSSInTsAndFLikePosets,
  author={Onak, Krzysztof and Parys, Pawel},
  booktitle={2006 47th Annual IEEE Symposium on Foundations of Computer Science (FOCS'06)}, 
  title={Generalization of Binary Search: Searching in Trees and Forest-Like Partial Orders}, 
  year={2006},
  volume={},
  number={},
  pages={379-388},
  doi={10.1109/FOCS.2006.32}}

@inproceedings{Mozes_Onak2008FindOptTSStartInLinTime,
  author       = {Shay Mozes and
                  Krzysztof Onak and
                  Oren Weimann},
  editor       = {Shang{-}Hua Teng},
  title        = {Finding an optimal tree searching strategy in linear time},
  booktitle    = {Proceedings of the Nineteenth Annual {ACM-SIAM} Symposium on Discrete
                  Algorithms, {SODA} 2008, San Francisco, California, USA, January 20-22,
                  2008},
  pages        = {1096--1105},
  publisher    = {{SIAM}},
  year         = {2008},
  url          = {http://dl.acm.org/citation.cfm?id=1347082.1347202},
  bibsource    = {dblp computer science bibliography, https://dblp.org}
}

@article{Lam1998ERankOfGsIsH,
title = {Edge ranking of graphs is hard},
journal = {Discrete Applied Mathematics},
volume = {85},
number = {1},
pages = {71-86},
year = {1998},
issn = {0166-218X},
doi = {https://doi.org/10.1016/S0166-218X(98)00029-8},
author = {Tak Wah Lam and Fung Ling Yue}
}

@InProceedings{dereniowski2022CFApproxAlgForBSInTsWithMonoQTimes,
  author =	{Dereniowski, Dariusz and Wrosz, Izajasz},
  title =	{{Constant-Factor Approximation Algorithm for Binary Search in Trees with Monotonic Query Times}},
  booktitle =	{47th International Symposium on Mathematical Foundations of Computer Science (MFCS 2022)},
  pages =	{42:1--42:15},
  series =	{Leibniz International Proceedings in Informatics (LIPIcs)},
  ISBN =	{978-3-95977-256-3},
  ISSN =	{1868-8969},
  year =	{2022},
  volume =	{241},
  editor =	{Szeider, Stefan and Ganian, Robert and Silva, Alexandra},
  publisher =	{Schloss Dagstuhl -- Leibniz-Zentrum f{\"u}r Informatik},
  address =	{Dagstuhl, Germany},
  URN =		{urn:nbn:de:0030-drops-168405},
  doi =		{10.4230/LIPIcs.MFCS.2022.42}
}

@misc{dereniowski2024SInTsMonoQTs,
      title={Searching in trees with monotonic query times}, 
      author={Dariusz Dereniowski and Izajasz Wrosz},
      year={2024},
      eprint={2401.13747},
      archivePrefix={arXiv},
      primaryClass={cs.DS},
      url={https://arxiv.org/abs/2401.13747}, 
}

@article{Angelidakis2018ShortestPQ,
  title={Shortest path queries, graph partitioning and covering problems in worst and beyond worst case settings},
  author={Haris Angelidakis},
  journal={ArXiv},
  year={2018},
  volume={abs/1807.09389},
  url={https://api.semanticscholar.org/CorpusID:51718679}
}

@article{DereniowskiVxRankOfChGsAndWTs,
title = {Vertex rankings of chordal graphs and weighted trees},
journal = {Information Processing Letters},
volume = {98},
number = {3},
pages = {96-100},
year = {2006},
issn = {0020-0190},
doi = {https://doi.org/10.1016/j.ipl.2005.12.006},
author = {Dariusz Dereniowski and Adam Nadolski}
}

@article{DereniowskiEfPQProcByGRank,
author = {Dereniowski, Dariusz and Kubale, Marek},
year = {2006},
month = {02},
pages = {273-285},
title = {Efficient Parallel Query Processing by Graph Ranking},
volume = {69},
journal = {Fundam. Inform.},
doi = {10.3233/FUN-2006-69302}
}

@article{DereniowskiERAndSInPOSets,
title = {Edge ranking and searching in partial orders},
journal = {Discrete Applied Mathematics},
volume = {156},
number = {13},
pages = {2493-2500},
year = {2008},
note = {Fifth International Conference on Graphs and Optimization},
issn = {0166-218X},
doi = {https://doi.org/10.1016/j.dam.2008.03.007},
author = {Dariusz Dereniowski}
}

@book{Knuth1973,
  author = {Knuth, Donald},
  pages = {391-392},
  publisher = {Addison-Wesley},
  title = {The Art Of Computer Programming, vol. 3: Sorting And Searching},
  year = 1973
}

@article{GIANNOPOULOU20122089,
title = {LIFO-search: A min–max theorem and a searching game for cycle-rank and tree-depth},
journal = {Discrete Applied Mathematics},
volume = {160},
number = {15},
pages = {2089-2097},
year = {2012},
issn = {0166-218X},
doi = {https://doi.org/10.1016/j.dam.2012.03.015},
url = {https://www.sciencedirect.com/science/article/pii/S0166218X12001199},
author = {Archontia C. Giannopoulou and Paul Hunter and Dimitrios M. Thilikos}
}

@article{NESETRIL20061022,
title = {Tree-depth, subgraph coloring and homomorphism bounds},
journal = {European Journal of Combinatorics},
volume = {27},
number = {6},
pages = {1022-1041},
year = {2006},
issn = {0195-6698},
doi = {https://doi.org/10.1016/j.ejc.2005.01.010},
url = {https://www.sciencedirect.com/science/article/pii/S0195669805000570},
author = {Jaroslav Nešetřil and Patrice {Ossona de Mendez}}
}

@techreport{Pothen1988OptimalEliminationTrees,
  author       = {Pothen, Alex},
  title        = {The Complexity of Optimal Elimination Trees},
  institution  = {Pennsylvania State University, Department of Computer Science},
  type         = {Technical Report},
  number       = {CS‑88‑16},
  month        = {April},
  year         = {1988},
  note         = {Penn State University CS‑88‑16}
}

@article{KATCHALSKI1995141,
title = {Ordered colourings},
journal = {Discrete Mathematics},
volume = {142},
number = {1},
pages = {141-154},
year = {1995},
issn = {0012-365X},
doi = {https://doi.org/10.1016/0012-365X(93)E0216-Q},
url = {https://www.sciencedirect.com/science/article/pii/0012365X93E0216Q},
author = {Meir Katchalski and William McCuaig and Suzanne Seager}
}

@inproceedings{Acostfunctionforsimilaritybasedhierarchicalclustering,
author = {Dasgupta, Sanjoy},
title = {A cost function for similarity-based hierarchical clustering},
year = {2016},
isbn = {9781450341325},
publisher = {Association for Computing Machinery},
address = {New York, NY, USA},
url = {https://doi.org/10.1145/2897518.2897527},
doi = {10.1145/2897518.2897527},
booktitle = {Proceedings of the Forty-Eighth Annual ACM Symposium on Theory of Computing},
pages = {118–127},
numpages = {10},
location = {Cambridge, MA, USA},
series = {STOC '16}
}

@inproceedings{Approximatehierarchicalclusteringviasparsestcutandspreadingmetrics,
author = {Charikar, Moses and Chatziafratis, Vaggos},
title = {Approximate hierarchical clustering via sparsest cut and spreading metrics},
year = {2017},
publisher = {Society for Industrial and Applied Mathematics},
address = {USA},
booktitle = {Proceedings of the Twenty-Eighth Annual ACM-SIAM Symposium on Discrete Algorithms},
pages = {841–854},
numpages = {14},
location = {Barcelona, Spain},
series = {SODA '17}
}

@InProceedings{noisyBSSFC,
  author =	{Dereniowski, Dariusz and {\L}ukasiewicz, Aleksander and Uzna\'{n}ski, Przemys{\l}aw},
  title =	{{Noisy (Binary) Searching: Simple, Fast and Correct}},
  booktitle =	{42nd International Symposium on Theoretical Aspects of Computer Science (STACS 2025)},
  pages =	{29:1--29:18},
  series =	{Leibniz International Proceedings in Informatics (LIPIcs)},
  ISBN =	{978-3-95977-365-2},
  ISSN =	{1868-8969},
  year =	{2025},
  volume =	{327},
  editor =	{Beyersdorff, Olaf and Pilipczuk, Micha{\l} and Pimentel, Elaine and Thang, Nguyen Kim},
  publisher =	{Schloss Dagstuhl -- Leibniz-Zentrum f{\"u}r Informatik},
  address =	{Dagstuhl, Germany},
  URL =		{https://drops.dagstuhl.de/entities/document/10.4230/LIPIcs.STACS.2025.29},
  URN =		{urn:nbn:de:0030-drops-228551},
  doi =		{10.4230/LIPIcs.STACS.2025.29}
}

@article{Dereniowski2024OnMG,
  title={On multidimensional generalization of binary search},
  author={Dariusz Dereniowski and Przemysław Gordinowicz and Karolina Wr'obel},
  journal={ArXiv},
  year={2024},
  volume={abs/2404.13193},
  url={https://api.semanticscholar.org/CorpusID:269293685}
}

@article{BOROWIECKI2023113682,
title = {The complexity of bicriteria tree-depth},
journal = {Theoretical Computer Science},
volume = {947},
pages = {113682},
year = {2023},
issn = {0304-3975},
doi = {https://doi.org/10.1016/j.tcs.2022.12.032},
url = {https://www.sciencedirect.com/science/article/pii/S0304397522007666},
author = {Piotr Borowiecki and Dariusz Dereniowski and Dorota Osula}
}

@article{Dereniowski2023Edge,
  author    = {Dariusz Dereniowski and Przemysław Gordinowicz and Paweł Prałat},
  title     = {Edge and Pair Queries—Random Graphs and Complexity},
  journal   = {The Electronic Journal of Combinatorics},
  volume    = {30},
  number    = {2},
  year      = {2023},
  article   = {P2.34},
  doi       = {10.37236/11159},
  url       = {https://www.combinatorics.org/ojs/index.php/eljc/article/view/v30i2p34}
}

@InProceedings{EfficientNoisyBinarySearch,
author="Dereniowski, Dariusz
and {\L}ukasiewicz, Aleksander
and Uzna{\'{n}}ski, Przemys{\l}aw",
editor="Flocchini, Paola
and Moura, Lucia",
title="An Efficient Noisy Binary Search in Graphs via Median Approximation",
booktitle="Combinatorial Algorithms",
year="2021",
publisher="Springer International Publishing",
address="Cham",
pages="265--281",
isbn="978-3-030-79987-8"
}

@inproceedings{kseparator,
  TITLE = {{The k-separator problem}},
  AUTHOR = {Ben-Ameur, Walid and Mohamed, Mohamed S. A. and Neto, Jos{\'e}},
  URL = {https://hal.science/hal-00843860},
  BOOKTITLE = {{COCOON '13 : 19th International Computing \& Combinatorics Conference}},
  ADDRESS = {Hangzhou, China},
  HAL_LOCAL_REFERENCE = {13303},
  PUBLISHER = {{Springer}},
  VOLUME = {7936},
  PAGES = {337-348},
  YEAR = {2013},
  MONTH = Jun,
  DOI = {10.1007/978-3-642-38768-5\_31},
  HAL_ID = {hal-00843860},
  HAL_VERSION = {v1},
}

@inproceedings{Improvedapproximationalgorithmsvertexseparators,
author = {Feige, Uriel and Hajiaghayi, MohammadTaghi and Lee, James R.},
title = {Improved approximation algorithms for minimum-weight vertex separators},
year = {2005},
isbn = {1581139608},
publisher = {Association for Computing Machinery},
address = {New York, NY, USA},
url = {https://doi.org/10.1145/1060590.1060674},
doi = {10.1145/1060590.1060674},
booktitle = {Proceedings of the Thirty-Seventh Annual ACM Symposium on Theory of Computing},
pages = {563–572},
numpages = {10},
location = {Baltimore, MD, USA},
series = {STOC '05}
}

@article{RankingsofGraphs,
author = {Bodlaender, Hans L. and Deogun, Jitender S. and Jansen, Klaus and Kloks, Ton and Kratsch, Dieter and M\"{u}ller, Haiko and Tuza, Zsolt},
title = {Rankings of Graphs},
journal = {SIAM Journal on Discrete Mathematics},
volume = {11},
number = {1},
pages = {168-181},
year = {1998},
doi = {10.1137/S0895480195282550},

URL = { 
    
        https://doi.org/10.1137/S0895480195282550
    
    

},
eprint = { 
    
        https://doi.org/10.1137/S0895480195282550
    
    

}
}

@misc{szyfelbein2025searchingtreeskupmodularweight,
      title={Searching in trees with k-up-modular cost functions}, 
      author={Michał Szyfelbein},
      year={2025},
      eprint={2504.17887},
      archivePrefix={arXiv},
      primaryClass={cs.DS},
      url={https://arxiv.org/abs/2504.17887}, 
}

@article{BODLAENDER1995238,
title = {Approximating Treewidth, Pathwidth, Frontsize, and Shortest Elimination Tree},
journal = {Journal of Algorithms},
volume = {18},
number = {2},
pages = {238-255},
year = {1995},
issn = {0196-6774},
doi = {https://doi.org/10.1006/jagm.1995.1009},
url = {https://www.sciencedirect.com/science/article/pii/S0196677485710097},
author = {H.L. Bodlaender and J.R. Gilbert and H. Hafsteinsson and T. Kloks}
}

@phdthesis{Berendsohn2024,
author = {Berendsohn, Benjamin Aram},
year = {2024},
title = {Search trees on graphs},
type = {Dissertation},
url = "http://dx.doi.org/10.17169/refubium-45704",
}

@article{OptimalSinT,
author = {Ben-Asher, Yosi and Farchi, Eitan and Newman, Ilan},
title = {Optimal Search in Trees},
journal = {SIAM Journal on Computing},
volume = {28},
number = {6},
pages = {2090-2102},
year = {1999},
doi = {10.1137/S009753979731858X},

URL = { 
    
        https://doi.org/10.1137/S009753979731858X
    
    

},
eprint = { 
    
        https://doi.org/10.1137/S009753979731858X
    
    

}
}

@InProceedings{Dereniowski2003CholeskyFactofMx,
author="Dereniowski, Dariusz
and Kubale, Marek",
editor="Wyrzykowski, Roman
and Dongarra, Jack
and Paprzycki, Marcin
and Wa{\'{s}}niewski, Jerzy",
title="Cholesky Factorization of Matrices in Parallel and Ranking of Graphs",
booktitle="Parallel Processing and Applied Mathematics",
year="2004",
publisher="Springer Berlin Heidelberg",
address="Berlin, Heidelberg",
pages="985--992",
isbn="978-3-540-24669-5"
}

@article{OnAGPartWAppVLSI,
title = {On a graph partition problem with application to VLSI layout},
journal = {Information Processing Letters},
volume = {43},
number = {2},
pages = {87-94},
year = {1992},
issn = {0020-0190},
doi = {https://doi.org/10.1016/0020-0190(92)90017-P},
url = {https://www.sciencedirect.com/science/article/pii/002001909290017P},
author = {Arunabha Sen and Haiyong Deng and Sumanta Guha}
}

@article{OnMinERSTs,
title = {On Minimum Edge Ranking Spanning Trees},
journal = {Journal of Algorithms},
volume = {38},
number = {2},
pages = {411-437},
year = {2001},
issn = {0196-6774},
doi = {https://doi.org/10.1006/jagm.2000.1143},
url = {https://www.sciencedirect.com/science/article/pii/S019667740091143X},
author = {Kazuhisa Makino and Yushi Uno and Toshihide Ibaraki}
}

@InProceedings{MinERSTrofTGs,
author="Makino, Kazuhisa
and Uno, Yushi
and Ibaraki, Toshihide",
editor="Bose, Prosenjit
and Morin, Pat",
title="Minimum Edge Ranking Spanning Trees of Threshold Graphs",
booktitle="Algorithms and Computation",
year="2002",
publisher="Springer Berlin Heidelberg",
address="Berlin, Heidelberg",
pages="428--440",
isbn="978-3-540-36136-7"
}

@techreport{ParAofModPs,
  author       = {A. V. Iyer and H. D. Ratliff and G. Vijayan},
  title        = {Parallel Assembly of Modular Products—An Analysis},
  institution  = {Georgia Institute of Technology},
  number       = {Tech.\ Report 88-06},
  year         = {1988},
  address      = {Atlanta, Georgia}
}

@article{HCObjFsandAlgs,
author = {Cohen-addad, Vincent and Kanade, Varun and Mallmann-trenn, Frederik and Mathieu, Claire},
title = {Hierarchical Clustering: Objective Functions and Algorithms},
year = {2019},
issue_date = {August 2019},
publisher = {Association for Computing Machinery},
address = {New York, NY, USA},
volume = {66},
number = {4},
issn = {0004-5411},
url = {https://doi.org/10.1145/3321386},
doi = {10.1145/3321386},
journal = {J. ACM},
month = jun,
articleno = {26},
numpages = {42}
}
\appendix

\section{Related work}\label{relatedwork}
The Graph Search Problem and its variants are related to multiple independently studied problems. These include, among others: 
\begin{itemize}
    \item Binary Search \cite{OnakParys2006GenOfBSSInTsAndFLikePosets,dereniowski2017ApproxSsForGeneralBSinWTs,Deligkas2019BsInGsRev,Emamjomeh2016DetAndProbBSinGs,dereniowski2022CFApproxAlgForBSInTsWithMonoQTimes,dereniowski2024SInTsMonoQTs,noisyBSSFC,Dereniowski2024OnMG,EfficientNoisyBinarySearch,Dereniowski2023Edge},
    \item Tree Search Problem \cite{Jacobs2010OnTheComplexSearchInTsAvg,Cicalese2014ImprovedApproxAvgTs,Cicalese2016OnTSPwNonUniCost}, 
    \item Binary Identification Problem \cite{Cicalese2012BinIdentPForWTs}, 
    \item Ranking Colorings \cite{Knuth1973,Dereniowski2009ERankOfWTs,DereniowskiERAndSInPOSets,DereniowskiEfPQProcByGRank,DereniowskiVxRankOfChGsAndWTs,Lam1998ERankOfGsIsH}, 
    \item Ordered Colorings \cite{KATCHALSKI1995141}, 
    \item Elimination Trees \cite{Pothen1988OptimalEliminationTrees}, 
    \item Hub Labeling \cite{Angelidakis2018ShortestPQ},
    \item Tree-Depth \cite{NESETRIL20061022,BOROWIECKI2023113682},
    \item Partition Trees \cite{Hgemo2024TightAB},
    \item Hierarchical Clustering \cite{Acostfunctionforsimilaritybasedhierarchicalclustering,HCObjFsandAlgs,Approximatehierarchicalclusteringviasparsestcutandspreadingmetrics}, 
    \item Search Trees on Trees \cite{SplayTonT,Fast_app_centroid_trees}, 
    \item LIFO-Search \cite{GIANNOPOULOU20122089}. 
\end{itemize}

All of the above definitions are equivalent when the input graph is a tree, but may differ for general graphs.

A variant of the Graph Search Problem can also be formulated in which queries 
are performed on edges rather than vertices. 
In what follows, we summarize the most important and relevant results, 
organized according to the query model and the objective function\footnote{Note that the notion of vertex weights is relevant only in the average-case setting.}.
\subsubsection{Vertex queries, worst-case cost}
When the input graph is a tree and all costs are uniform, the problem is solvable in linear time \cite{Schaffer1989OptNodeRankOfTsInLinTime,OnakParys2006GenOfBSSInTsAndFLikePosets} and $O\br{\log n}$ queries always suffice. Beyond trees, the problem is known to be NP-hard even in: chordal graphs \cite{DereniowskiVxRankOfChGsAndWTs}, bipartite and co-bipartite graphs \cite{RankingsofGraphs}. By combining the results of \cite{BODLAENDER1995238} and \cite{Improvedapproximationalgorithmsvertexseparators} one can obtain a general $O\br{\log^{3/2}n}$-approximation algorithm. Additionally, the problem is solvable in polynomial time for graphs with bounded treewidth \cite{RankingsofGraphs}. 

In the case of non-uniform costs, the problem is known to be NP-hard even for trees 
\cite{Dereniowski2009ERankOfWTs,Cicalese2012BinIdentPForWTs,Cicalese2016OnTSPwNonUniCost}, 
for which there exists an $O\br{\sqrt{\log n}}$-approximation algorithm 
\cite{dereniowski2017ApproxSsForGeneralBSinWTs}. 
Further improvements are possible for restricted classes of the cost function. 
These results include:
\begin{itemize}
    \item A parametrized PTAS running in $O\br{\br{cn/\epsilon}^{2c/\epsilon}}$ time, where $c$ is the largest cost \cite{DereniowskiVxRankOfChGsAndWTs}, 
    \item A 2-approximation algorithm for down-monotonic cost functions \cite{dereniowski2022CFApproxAlgForBSInTsWithMonoQTimes},
    \item An 8-approximation algorithm for up-monotonic cost functions \cite{dereniowski2022CFApproxAlgForBSInTsWithMonoQTimes,dereniowski2024SInTsMonoQTs},
    \item A parametrized $O\br{\log\log n}$-approximation algorithm for $k$-up-modular cost functions, running in $k^{O\br{\log k}}\cdot\text{poly}\br{n}$ time \cite{szyfelbein2025searchingtreeskupmodularweight}.
\end{itemize}

For general graphs with non-uniform costs, the problem can be reduced to the uniform-cost case 
by rounding all costs to polynomial values and applying a reduction similar to that in 
\cite{DereniowskiVxRankOfChGsAndWTs}, which yields an $O\br{\log^{3/2} n}$-approximation.

\subsubsection{Vertex queries, average-case cost}

For trees with non-uniform weights and uniform query costs, a PTAS running in 
$O\br{n^{2/\epsilon+1}}$ time exists \cite{SplayTonT}. 
This was later improved to an FPTAS with running time 
$O\br{\br{1/\epsilon}^{2/\log_23}\cdot n^{1+4/\log_23}\cdot \log^2\br{n/\epsilon}}$ 
\cite{Berendsohn2024}. 
Every optimal decision tree has height at most $O\br{\log w\br{T}}$, and a simple 
decision tree that always queries the weighted centroid achieves a 2-approximation 
\cite{Fast_app_centroid_trees}. 
Beyond trees, the problem is NP-hard even for graphs with treewidth at most 15, 
as well as for dense graphs with uniform weights \cite{Berendsohn2024}. 

\subsubsection{Edge queries, worst-case cost}

For trees with uniform costs, the problem is solvable in linear time 
\cite{Lam1998ERankOfGsIsH,Mozes_Onak2008FindOptTSStartInLinTime}, 
and at most 
$\frac{\Delta\br{T}-1}{\log\br{\Delta\br{T}+1}-1}\cdot \log n$ queries always suffice 
\cite{Emamjomeh2016DetAndProbBSinGs}. 
For general graphs, the problem is known to be NP-hard \cite{Lam1998ERankOfGsIsH}. 
When query costs are non-uniform, an $O\br{\log n}$-approximation was given 
\cite{Dereniowski2009ERankOfWTs}, later improved to 
$O\br{\log n/\log\log\log n}$ \cite{Cicalese2012BinIdentPForWTs}, then to 
$O\br{\log n/\log\log n}$, and finally, by reduction from the vertex version of the 
problem, to $O\br{\sqrt{\log n}}$ \cite{dereniowski2017ApproxSsForGeneralBSinWTs}.
\subsubsection{Edge queries, average-case cost}

For trees with uniform costs, the problem is known to be (weakly) NP-hard 
even for trees with diameter at most $4$ and for trees with degree at most $16$ 
\cite{Jacobs2010OnTheComplexSearchInTsAvg}. 
Every optimal decision tree has height at most $O\br{\Delta \log w\br{T}}$, and there exists 
an FPTAS running in $\text{poly}\br{n^{\Delta\br{T}}/\epsilon}$ time. 
A simple greedy algorithm that always queries the edge which splits the weights most evenly 
achieves a 2-approximation \cite{Jacobs2010OnTheComplexSearchInTsAvg}, later improved to $\phi$ in \cite{Cicalese2014ImprovedApproxAvgTs} and $3/2$ in 
\cite{Hgemo2024TightAB} which is tight. 

For uniform weights and non-uniform costs, there exists an $6.75$-approximation algorithm for trees 
and an $O\br{\sqrt{\log n}}$-approximation algorithm for general graphs 
\cite{HCObjFsandAlgs,Approximatehierarchicalclusteringviasparsestcutandspreadingmetrics}. 
On the hardness side, the problem cannot be approximated within any constant factor 
if the \textit{Small Set Expansion Hypothesis} holds 
\cite{Approximatehierarchicalclusteringviasparsestcutandspreadingmetrics}, even when all costs are uniform.

\section{Hardness of the Tree Search Problem}\label{hardness}
In the decision version of the Tree Search Problem, one is asked to determine whether, 
for a given instance $\br{T, c, w}$, there exists a decision tree of cost at most $K$.
\begin{theorem}
    The Decision Tree Search Problem is NP-complete even when restricted to trees with $\Delta\br{T}\leq 16$ and to trees with $\diam\br{T}\leq 8$.
\end{theorem}

    \begin{proof}
        The problem is in NP since, given a decision tree $D$, one can verify in polynomial time 
whether all the requirements are satisfied.

To show hardness, we use a black-box reduction from the edge-query, uniform-cost, 
and non-uniform-weight variant, which is NP-complete even when restricted to trees with $\Delta\br{T}\leq 16$ and to trees with $\diam\br{T}\leq 4$ \cite{Jacobs2010OnTheComplexSearchInTsAvg}. 
Let $\br{T, w, K}$ be such an instance. 
We construct a new instance $\br{T', c, w, K'}$ for the TSP as follows: 
for every $v \in V\br{T}$, we set $c\br{v} = K+1$. 
We subdivide each edge $e \in E\br{T}$ by adding a new vertex $v_e$ with 
$w\br{v_e} = 0$ and $c\br{v_e} = 1$. 
We set $K' = K + w\br{T} \cdot \br{K+1}$.

Assume that we have a decision tree $D$ of cost at most $K$ for the original instance. 
To obtain a decision tree $D'$ for the new instance, we replace each query in $D$ 
with a query to the vertex that subdivides the corresponding edge. 
Additionally, below each leaf of $D$, we attach the appropriate queries to the original vertices. 
As $D$ contains a query to every edge of $T$, each vertex is separated, so for every 
$v \in V\br{T}$, one such additional query is added. 
This results in a decision tree $D'$ of cost at most $K + w\br{T} \cdot \br{K+1}= K'$, 
as required.

Observe that in the new instance, for every vertex $v \in T$, the cost of searching for $v$ 
is at least $K+1$ since $c\br{v} = K+1$. 
Therefore, for these vertices, at least $w\br{T} \cdot \br{K+1}$ cost is required. 
This implies that each such vertex has exactly one such query in its query sequence, 
namely the query to $v$ itself. Otherwise, the cost would exceed $K'$, and we conclude 
that every such $v$ is queried only when the candidate subset consists solely of $v$. 

Conversely, assume there exists a decision tree $D'$ of cost at most $K'$ for the new instance. 
We show how to obtain a decision tree $D$ for the original instance. 
We replace each query to a vertex $v_e$ with a query to edge $e$, 
and delete all queries to vertices $v \in V\br{T}$. 
Since each $v \in V\br{T}$ was the last query in $Q_{D'}\br{T', v}$, performed 
when the candidate set consisted only of $v$, the resulting $D$ is a valid decision tree 
for the original instance. Additionally, the cost of $D$ is at most $K' - w\br{T} \cdot \br{K+1} = K$, 
as required.

Finally, note that subdividing each edge doubles the diameter of the tree while leaving the degree unchanged, 
so the claim follows.

    \end{proof}

\section{Pseudoexact algorithm for the Weighted $\alpha$-Separator Problem}\label{separatorFPTAS}

We devise a dynamic programming procedure similar to the one in \cite{kseparator} and combine it with a rounding trick to obtain a bi-criteria FPTAS.
Note that the authors considered only the case in which all weights are uniform. However, we generalize their algorithm to arbitrary integer weights and introduce an additional case that 
was previously lacking\footnote{Probably due to an oversight.}.

\begin{theorem}\label{separator}
    Let $T$ be a tree. 
    There exists an optimal algorithm for the Weighted $\alpha$-Separator Problem running in 
    $O\br{n\cdot \br{w\br{T}/\alpha}^2}$ time.
\end{theorem}

\begin{proof}
    Assume that the input tree is rooted at an arbitrary vertex $r\br{T}$. Let $k=\fl{w\br{T}/\alpha}$. We want to find a separator $S$ such that for every $H\in T-S$, $w\br{H}\leq k$.
    Let $C_{v}$ denote the cost of the optimal separator $S_v$ in $T_v$ with this property. 
    Define $C_{v}^{in}$ as the cost of the optimal separator for $T_v$, 
    under the condition that $v \in S_v$. 
    We immediately have:
    $$
    C_{v}^{in} = c\br{v}+\sum_{c\in \mathcal{C}_{T,v}}C_{c}.
    $$
    
    Assume that $v \notin S_v$. 
    Let $H_v \in T_v - S_v$ be the component containing $v$. 
    For every integer $0 \leq w \leq k$, let $C_v^{out}\br{w}$ be the cost of the optimal separator for $T_v$, 
    such that $v \notin S_v$ and $w\br{H_v} = w$. 
    Then:
    $$
    C_v = \min \brc{C_{v}^{in}, \min_{0 \leq w \leq k} C_v^{out}\br{w}}.
    $$
    
    For any vertex $v \in V\br{T}$ and any integer $1 \leq i \leq \deg_{T,v}^+$, 
    let $S_{v,i}$ be the optimal separator for $T_{v,i}$ and $H_{v,i} \in T_{v,i} - S_{v,i}$ be the component containing $v$.
    For any integer $0 \leq w \leq k$, let $C_{v,i}^{out}\br{w}$ be the cost of an optimal separator for $T_{v,i}$, 
    such that $v \notin S_{v,i}$ and $w\br{H_{v,i}} = w$. 
    Then
    $$
    C_{v}^{out}\br{w} = C_{v,\deg_{T,v}^+}^{out}\br{w}.
    $$
    
    For $i = 1$ we have:
    $$
    C_{v,1}^{out}\br{w} =
    \begin{cases}
        \infty, & \text{if } w < w\br{v},\\
        \min \brc{C_{c_1}^{in},C_{c_1}^{out}\br{0}}, & \text{if } w = w\br{v},\\
        C_{c_1}^{out}\br{w - w\br{v}}, & \text{if } w > w\br{v}.
    \end{cases}
    $$
    
    For $i > 1$:
    $$
    C_{v,i}^{out}\br{w} = 
    \min \brc{
        C_{v,i-1}^{out}\br{w} + C_{c_i}^{in}, 
        \min_{0 \leq j \leq w} \brc{ C_{v,i-1}^{out}\br{w-j} + C_{c_i}^{out}\br{j} }
    }.
    $$
    
    In the above, the first term of the outer minimum corresponds to the case $c_i \in S_{v,i}$, 
    so $H_{v,i} = H_{v,i-1}$. 
    The second term considers the alternative, checking all possible partitions of 
    the weight between $H_{v,i-1}$ and $H_{c_i}$.

    These relationships suffice to compute $C_{r\br{T}}$, the cost of the optimal separator $S$ for $T$. 
    Computation is performed in a bottom-up, left-to-right manner, starting from the leaves. 
    For a leaf $v$, we have $C_v^{in} = c\br{v}$ and:
    $$C_v^{out}\br{w} = \begin{cases}
        0, & \text{if } w = w\br{v}\leq k,\\
        \infty, & \text{otherwise.}
    \end{cases}$$ 

    Since each of the $C_v^{in}$ subproblems requires $O\br{\deg_{T,v}^+}$ computational steps we get that they require $O\br{n}$ running time. As there are $O\br{n\cdot k} = O\br{n \cdot w\br{T}/\alpha}$ remaining subproblems and each requires 
    at most $O\br{k}=O\br{w\br{T}/\alpha}$ computational steps, the running time is $O\br{n \cdot \br{w\br{T}/\alpha}^2}$.
\end{proof}

Note that, the running time of the above procedure depends on $w\br{T}$ which may not be polynomial. 
To alleviate this difficulty, we slightly relax the condition on the size of components 
in $T-S$ using a controlled parameter $\delta$. 
Based on this relaxation, we show how to construct a bicriteria FPTAS for the problem. Let $\delta>0$ be any fixed constant and let $\FSeparator$ be the dynamic programming procedure from Theorem \ref{separator}.
        The algorithm is as follows:
        
\begin{algorithm}[H]
\caption{The bicriteria FPTAS for the Weighted $\alpha$-separator Problem}
\label{createDecisionTree}
\SetKwFunction{FDecisionTree}{DecisionTree}
\SetKwProg{Fn}{Procedure}{:}{}
\Fn{$\FSeparatorFPTAS\br{T, c, w, \alpha, \delta}$}{
$K\gets\frac{\delta\cdot w\br{T}}{n\cdot \alpha}$.

\ForEach{ $v\in V\br{T}$}
{$w'\br{v} \gets \fl{\frac{w\br{v}}{K}}$.}

$\alpha'\gets\frac{\alpha\cdot K\cdot w'\br{T}}{w\br
T}$.

$S'\gets\FSeparator\br{T, c, w', \alpha'}$.

\Return $S'$.
}
\end{algorithm}
        \begin{lemma}
            Let $S$ be the optimal separator for the $\br{T, c, w, \alpha}$ instance. We have that $c\br{S'}\leq c\br{S}$.
        \end{lemma}
            \begin{proof}
                We prove that $S$ is a valid separator for the $\br{T, c, w', \alpha'}$ instance, so that $c\br{S'}\leq c\br{S}$.
                To simplify the analysis, we will define the auxiliary instance: For every $v\in V\br
                {T}$, let $w''\br{v} = K\cdot\fl{\frac{w\br{v}}{K}} $. Additionally, let $ \alpha'' =\frac{\alpha \cdot w''\br{T}}{w\br{T}}$.
                
                In this new instance, for $v\in V\br{T}$ we have $w''\br{v}\leq w\br{v}$, so for every $H\in T-S$, 
                $$w''\br{H}\leq w\br{H}\leq w\br{T}/\alpha= w''\br{T}/\alpha''$$
                
                where the second inequality is by the definition of the $\alpha$-separator and the equality is by the definition of $\alpha''$.
                
                We conclude that $S$ is an $\alpha''$-separator for the auxiliary instance $\br{T, c, w'', \alpha''}$. Now notice that the $\br{T, c, w', \alpha'}$ instance has all of its weights scaled by a constant value of $K$, relatively to $\br{T, c, w'', \alpha''}$ and $\alpha' = \alpha''$. As multiplying weights by a constant does not influence the validity of a solution, $S$ is an $\alpha'$-separator for $\br{T, w', c, \alpha'}$ and the claim follows.
            \end{proof}
        \begin{lemma}
            For every $H\in T-S'$, we have that $w\br{H}\leq\frac{\br{1+\delta}\cdot w\br{T}}{\alpha}$.
        \end{lemma}
        
            \begin{proof}
                By definition $ \frac{w\br{v}}{K}\leq w'\br{v}+1$ and therefore, also $w\br{v}\leq K\cdot w'\br{v}+K$. We have:
                \begin{align*}
                \sum_{v\in H}w\br{v}&\leq K\cdot\sum_{v\in H}w'\br{v}+K\cdot n\leq \frac{K\cdot w'\br{T}}{\alpha'}+K\cdot n \\
                &= \frac{w\br{T}}{\alpha} + \frac{\delta \cdot w\br{T}}{\alpha}=\frac{\br{1+\delta}\cdot w\br{T}}{\alpha}
                \end{align*}
                
                where the second inequality is due to the fact that $S'$ is a $\alpha'$-separator for $\br{T, c, w', \alpha'}$ instance and the first equality is by the definition of $\alpha'$ and $K$.
            \end{proof}
            
        Combining the two above lemmas with the fact that $\frac{w'\br{T}}{\alpha'}=\frac{w\br{T}}{K\cdot \alpha}=n/\delta$ we have that the algorithm runs in time $O\br{n^3/\delta^2}$ as required.

\section{Other omitted Proofs}\label{omittedproofs}

\subsection{Proof of Lemma \ref{neighborsPathLemma}}\label{proof_neighborsPathLemma}
    \begin{proof}
        Let $q$ be any query in $D$. There are two cases:
        \begin{enumerate}
            \item $q\in V\br{G-V\br{\mathcal{G}}-N_{G}\br{V\br{\mathcal{G}}}}$. Then, for every $x\in N_{G}\br{V\br{\mathcal{G}}}$ being the target, the answer is the same connected component of $G-q$. Therefore, $q$ has at most one child $u$ in $D$, such that $V\br{D_u}\cap Q \neq \emptyset$.
            \item $q\in N_{G}\br{V\br{\mathcal{G}}}$. After a query to $q$, the situation is the same as in the first case, except when $x=q$. Then the response is $x$ itself, in which case no further queries are needed, and again $q$ has at most one child $u$ in $D$, such that $V\br{D_u}\cap Q \neq \emptyset$.
        \end{enumerate}
    \end{proof}
    
\subsection{Proof of Lemma \ref{lambda_lemma}}\label{proof_lambda_lemma}
            \begin{proof}
                There are two cases:
                \begin{enumerate}
                    \item $w\br{ S_{\mathcal{G}}^*}\geq w\br{\mathcal{G}}/\sqrt{\lambda}$. In this case we take arbitrary partition $\mathcal{A}, \mathcal{B}$ of $\mathcal{H}$. We have:
                    $$w\br{A\cup S_{\mathcal{G}}^*}\cdot w\br{B\cup S_{\mathcal{G}}^*}\geq w\br{ S_{\mathcal{G}}^*}^2 \geq w\br{\mathcal{G}}^2/\lambda.$$
                \item $w\br{ S_{\mathcal{G}}^*} \leq w\br{\mathcal{G}}/\sqrt{\lambda}$.
                For any choice of the partition $\mathcal{A},\mathcal{B}$ of $\mathcal{H}$, we have $\frac{w\br{A\cup B}}{w\br{\mathcal{G}}}\geq 1-\frac{1}{\sqrt{\lambda}}$. We pick $\mathcal{A},\mathcal{B}$ to be a partition of $\mathcal{H}$, such that $w\br{A}\geq w\br{B}\geq \br{\frac{1}{2}-\frac{1}{\sqrt{\lambda}}}\cdot w\br{\mathcal{G}}$ (this is always possible as $\frac{1}{2}-\frac{1}{\sqrt{\lambda}}>0$ and for each $H\in\mathcal{H}$, $w\br{H}\leq w\br{\mathcal{G}}/2$). We have:
                \begin{align*}
                w\br{A\cup S_{\mathcal{G}}^*}\cdot w\br{B\cup S_{\mathcal{G}}^*}&\geq w\br{A}\cdot w\br{B} \\&\geq \br{\br{1-{1}/{\sqrt{\lambda}}}\cdot w\br{\mathcal{G}}-w\br{B}}\cdot w\br{B}\\&\geq  {w\br{\mathcal{G}}^2}/{2}\cdot \br{{1}/{2}-{1}/{\sqrt{\lambda}}} = {w\br{\mathcal{G}}^2}/{\lambda}    
                \end{align*}

                where the third inequality is by using the fact that the concave function $f\br{w\br{B}}=w\br{B}\cdot\br{\br{1-\frac{1}{\sqrt{\lambda}}}\cdot w\br{\mathcal{G}}-w\br{B}}$ reaches its minimum in the interval $\left[\br{\frac{1}{2}-\frac{1}{\sqrt{\lambda}}}\cdot w\br{\mathcal{G}},w\br{\mathcal{G}}/2\right]$ when $w\br{B}=\br{\frac{1}{2}-\frac{1}{\sqrt{\lambda}}}\cdot w\br{\mathcal{G}}$.
                \end{enumerate}
            \end{proof}
            \subsection{Proof of Lemma \ref{up_graphs}}\label{proof_of_up_graphs}
            \begin{proof}
                Fix a value of $\mathcal{G}$ and $k$. Their contribution to the cost is $c\br{S_{\fl{k/2}}^*\cap \mathcal{G}}$. Consider which candidate subgraphs contribute such a term. By definition of $S_{\mathcal{G}}$, we have that $\mathcal{G}$ is the minimal subgraph, such that $w\br{\mathcal{G}}\geq k\geq w\br{A_{\mathcal{G}}}+1 >w\br{H}$, for every $H\in \mathcal{G}-S_{\mathcal{G}}$. This means that if for every $H\in \mathcal{G}-S_{\mathcal{G}}$, $w\br{H}<k$, then $\mathcal{G}$ contributes such a term. Since for all $H_1, H_2\in \mathcal{G}-S_{\mathcal{G}}$, $H_1\cap H_2=\emptyset$, we have that $\br{S_{\fl{k/2}}^*\cap H_1}\cup \br{S_{\fl{k/2}}^*\cap H_2}=\emptyset$, the claim follows by summing over all values of $k$.
            \end{proof}
\end{document}